\documentclass[reqno]{amsart}
\usepackage{amsmath, amssymb, amsfonts, amsthm}
\newtheorem{theorem}{Theorem}[section]

\theoremstyle{definition}

\newtheorem{definition}[theorem]{Definition}

\newtheorem{cor}[theorem]{Corollary}

\newtheorem{prop}[theorem]{Proposition}

\theoremstyle{remark}

\newtheorem{remark}[theorem]{Remark}
\newcommand{\bes}{{\begin{split}}}
\newcommand{\ees}{{\end{split}}}
\newcommand{\bees}{{\begin{equation}\begin{split}}}
\newcommand{\es}{{\end{split}\end{equation}}}
\newcommand{\erre}{{\mathbb R}}

\newcommand{\bea}{\begin{eqnarray}}

\newcommand{\eea}{\end{eqnarray}}

\newcommand{\be}{\begin{equation}}

\newcommand{\ee}{\end{equation}}

\newcommand{\n}{\noindent}

\newcommand{\f}{\frac}

\newcommand{\ve}{\varepsilon}

\newcommand{\om}{\omega}

\newcommand{\mc}{\mathcal}

\numberwithin{equation}{section}

\usepackage{graphicx}
\begin{document}
\large

\title{On the structure of critical energy levels \\ for the cubic focusing NLS on
  star graphs}


\author{Riccardo Adami}

\address{Adami and Noja: Dipartimento di Matematica e Applicazioni, 
Universit\`a di Milano Bicocca, via R. Cozzi 53, 20125 Milano, Italy
and Istituto di Matematica Applicata e Tecnologie Informatiche ``Enrico Magenes'', CNR, via Ferrata, 1
27100 Pavia, Italy \\
e-mail: riccardo.adami@unimib.it, diego.noja@unimib.it}

\address{
Cacciapuoti: Hausdorff Center for Mathematics, Institut f\"ur
Angewandte Mathematik, 60, Endenicher Allee, 53115 Bonn, Germany \\
e-mail: cacciapuoti@him.uni-bonn.de}

\address{
Finco:  Facolt\`a di Ingegneria, Universit\`a Telematica
Internazionale Uninettuno, 
Corso Vittorio Emanuele II 39, 00186 Roma,
Italy \\
e.mail: d.finco@uninettunouniversity.net}




\author{Claudio Cacciapuoti}

\author{Domenico Finco}

\author{Diego Noja}

\date{February 11, 2012}



\begin{abstract} We provide information on a non trivial structure of
  phase space of the cubic NLS on a three-edge star graph.  
We prove that, contrarily to the case of the standard NLS on the line,
the energy associated to the cubic focusing Schr\"odinger equation on
the three-edge star graph with a free (Kirchhoff) vertex does not
attain a minimum value on any sphere of constant $L^2$-norm. We
moreover show that the only stationary state with prescribed
$L^2$-norm is indeed a saddle point.
\end{abstract}

\maketitle

\section{Introduction}



A major issue in nonlinear dynamics consists in the search for
stationary solutions and in the study of their stability
properties. In the case of hamiltonian systems, a first picture of the
phase space can be drawn by identifying critical points of the energy and their
nature. In particular it is important to know if a {\em ground state}
exists, where a ground state is defined as the minimizer of the energy
functional, possibly restricted to suitable submanifolds. In the
context of the nonlinear Schr\"odinger (NLS) equation, one typically
restricts the energy functional to a manifold on which a second
conserved quantity (sometimes called mass, or charge) is
constant. From a physical point of view, such a restriction is
meaningful, as the extra conserved quantity often represents a
physical characteristic of the system (e.g., the mass, or the number
of particles). In the one dimensional case, the ground states of the
NLS with power nonlinearity on the line are well known and completely
described in the classical paper \cite{BL1}, where more general
nonlinearities are also treated. It turns out that on the line, the NLS
energy constrained to the manifold of the states of constant mass,
attains its minimum value in correspondence of a unique (up to
translation) positive, symmetric and decreasing (for $x>0$)
function. No other critical points of the energy exist. \par\noindent
In this paper we are interested in the case of the focusing nonlinear
cubic Schr\"odinger equation on a three-edge star graph, sometimes called
in the physical literature a Y junction. To put the issue in a
physical context, we begin by recalling that the {\it linear}
Schr\"odinger equation on a graph is a well developed subject, as an
effective description of dynamics of many mesoscopic systems such as,
for example, quantum nanowires (see \cite{Ex, kostrykin, Kuchment} and
references therein). It is of interest to extend the analysis to the
nonlinear wave propagation on networks. In particular, it is well
known that the NLS appears as an effective equation in several
different areas: the description of Bose condensates, the propagation
of electromagnetic pulses in nonlinear (Kerr) media and Langmuir waves
in plasma physics. In many situations it seems of interest to treat
the propagation of NLS solutions associated to such phenomena in one
dimensional ramified structures, the prototype of which is a 
three-edge star graph, or Y-junction. The subject is at its beginnings and
for some preliminary experimental, numerical and analytical works see
\cite{Smi, Kevr, Miro, Sob, TOD, Wan}.  A first rigorous analysis of
nonlinear stationary (or bound) states for NLS on a star graphs is
outlined in \cite{ACFN2}, where solitary waves for a star graph with
delta vertex are constructed, including the case of a free or
Kirchhoff vertex. The Kirchhoff vertex is the closest analogue to the
free particle on the line, to which it reduces when the line is
considered as a graph with two edges. In the linear case the Kirchhoff
laplacian on a graph, in analogy with the laplacian on the line, has
only absolutely continuous spectrum, and so only scattering states are
possible for the Schr\"odinger dynamics. However, in the presence of a nonlinearity a
three-edge star graph with Kirchhoff conditions at the vertex, admits
a unique stationary state, which is quite simply described: it is the
state on the graph which coincides on every edge with half a
soliton. One could suspect that this stationary state is the ground
state. On the contrary, we show that, perhaps unexpectedly, 
this is not the case.

This fact highlights that the NLS on graphs, even on
the simplest one, exhibits remarkable differences
w.r.t. the same evolution equation on the line.

In this letter we
consider the case of the cubic focusing NLS and analyze the
energy constrained to a fixed sphere in $L^2$, i.e. the set of states
of prescribed mass. We show that the constrained energy is bounded from below (a fact that can be shown similarly
to the case of ${\erre}^n$ and subcritical nonlinearity, see \cite{C} and \cite{ACFN2}), but it approaches the
infimum without attaining a minimum value.
 The more, the only
nonlinear stationary state is a saddle point of the constrained energy
functional. The existence of a saddle point of the energy is a
remarkable feature, inducing on the phase space of the system stable
and unstable manifolds which, in the absence of other critical points,
are the main structural properties of the dynamics. The consequences
of this fact will be investigated in a subsequent paper. \par\noindent
\vskip5pt 
\section{Preliminaries}

Before giving the main results, we start by fixing the
framework, the notation, and recalling some basic results.

\medskip

\n
1. A three-edge star graph ${\mathcal G}$ can be thought of as composed
by  
three halflines with a
common origin, called {\em vertex}. A {\em state} or {\em
  wavefunction} on the graph is an element of the Hilbert space
$L^2 ({\mathcal G}) \ = \ \bigoplus_{i=1}^3 L^2 (\erre^+; dx_i)$, and
can be represented as a
column vector,
namely
$$
\Psi \ = \ \left( \begin{array}{c} \Psi_1 \\
\Psi_2 \\ \Psi_3
\end{array} \right), \qquad \Psi_i \in L^2 (\erre^+).
$$
The space $L^2 ({\mathcal G})$ is naturally endowed with the hermitian
product
$$ ( \Phi, \Psi)_{L^2 ({\mathcal G})} \ = \ \sum_{i=1}^3 \int_0^{+
  \infty} \overline{\Phi_i (x_i)}\Psi_i (x_i) \, dx_i.
$$
Sobolev  
and $L^p$-spaces 
on ${\mathcal G}$ are defined analogously, namely
$$ H^s ({\mathcal G}) \ = \ \ \bigoplus_{i=1}^3 H^s (\erre^+),
\qquad  L^p ({\mathcal G}) \ = \ \ \bigoplus_{i=1}^3 L^p (\erre^+).
$$
In the following, we denote by $\| \Psi \|_p$ the norm of the function
$\Psi$
in the space $L^p ({\mathcal G})$.
When $p=2$ we shall simply write $\| \Psi \|$.

\medskip

\noindent
2. The dynamics of the system is described by the Schr\"odinger
equation
\be \label{schrod}
i \partial_t \Psi (t) \ = \ - \Delta \Psi (t) - |\Psi (t)|^2\Psi (t),
\ee
where:
\begin{itemize}
\item The operator $-\Delta$ acts on the domain
\be \nonumber 
D (- \Delta ) \  : = \ \{  \Psi \in H^2 ({\mathcal G}), \ \Psi_1 (0) =
\Psi_2 (0)  =  \Psi_3 (0), \ \Psi^{\prime}_1 (0) +\Psi^{\prime}_2 (0) + \Psi^{\prime}_3 (0)  = 
0 \}
\ee
and its action reads 
$$ - (\Delta \Psi)_i \ = \ - \Psi_i^{\prime \prime}, \qquad i =
1,2,3. $$ The condition at the vertex is usually referred to as the
{\em Kirchhoff's boundary condition}. The operator $- \Delta$ is
selfadjoint on $L^2 ({\mathcal G})$.  Notice that, on the graph, 
the laplacian with
Kirchhoff boundary conditions is the natural generalization 
of the free laplacian on the line, as it is easily seen by considering
the line as a two-edge star graph, and noticing that the
boundary conditions reduces to continuity of the wavefunction and
continuity of the derivative at the vertex, i.e. $\Psi\in H^2(\erre)$.
 
\item The nonlinear term in \eqref{schrod} is defined componentwise, namely
$$ (|\Psi|^2 \Psi)_i : =  |\Psi_i|^2 \Psi_i.$$
\end{itemize}

\medskip

\n
3. The problem \eqref{schrod} is globally well-posed in $H^1
({\mathcal G})$ (see \cite{ACFN1}). The $L^2$-norm and the energy
\be \begin{split} \label{energy}
E (\Psi) & \ = \  \f 1 2 \| \Psi^\prime \|^2 - \f 1 4 \| \Psi \|_4^4 
 \ = \ \sum_{i=1}^3 \left( \f 1 2 \| \Psi_i^\prime \|_{L^2(\erre^+)}^2 - 
\f 1 4 \| \Psi_i \|_{L^4(\erre^+)}^4 \right)
\end{split} \ee
are conserved by time evolution. 

In the sequel we use the following notation:
\be \label{energy1}
 E_1 (\psi)  \ = \  \f 1 2 \| \psi^\prime \|_{L^2(\erre^+)}^2 - \f 1 4
\| \psi \|_{L^4(\erre^+)}^4, \qquad \psi \in H^1 (\erre^+)\ee 
and
\be \label{energy2} E_2 (\psi)  \ = \  \f 1 2 \| \psi^\prime \|_{L^2(\erre)}^2 - \f 1 4
\| \psi \|_{L^4(\erre)}^4,  \qquad \psi \in H^1 (\erre). 
\ee

\medskip

\n
4.
Let us recall two well-known results on minimization for the cubic
NLS on one- and two-edge star graphs (i.e., on the halfline and on the line):
\begin{enumerate}
\item The minimum of the functional $E_1$ on the functions in $H^1
  (\erre^+)$ with squared $L^2$-norm equal to $m > 0$ is achieved on
  the function (up to a phase factor)
$$ \phi_m (x) \ = \ \f m {\sqrt 2} \cosh^{-1} \left( \f m 2 x \right),
\qquad x \geq 0,$$  
and gives
\be \label{min1}
E_1 (\phi_m ) \ = \ - \f {m^3} {24}. 
\ee
\item 
The minimum of the functional $E_2$ on the functions in $H^1
  (\erre)$ with squared $L^2$-norm equal to $m > 0$ is achieved on
  the functions (up to a phase factor)
\be \label{solit} \varphi_m^y (x) \ = \ \f m {2 \sqrt 2} \cosh^{-1} \left( \f m 4 (x-y)
\right), \qquad x \in \erre 
\ee  
and gives
\be \label{min2}
E_2 (\varphi_m^y ) \ = \ - \f {m^3} {96}. 
\ee
\end{enumerate}

\section{Results}

As recalled in the introduction the energy functional \eqref{energy} at fixed $L^2$ norm is bounded from below for subcritical nonlinearity, which is our case.\par\noindent
Let us introduce the following family of states:
\begin{definition}
We call {\em sesquisoliton} (i.e. ``one and half" soliton) any function of the form
\be \label{sesqui}
\Phi_{m_1, m_2}^x (x_1,x_2,x_3) : = \left(
\begin{array} {c}  \f{m_1}{\sqrt 2} \cosh^{-1} \left(\f{m_1} 2 x_1\right) \\
\f{m_2}{2 \sqrt 2} \cosh^{-1} \left(\f{m_2}{4 } (x_2-x)\right) \\ 
\f{m_2}{2 \sqrt 2} \cosh^{-1} \left(\f{m_2}{4 } (x_3+x)\right)
\end{array} \right)
\ee
where $0 < m_1 \leq m_2$, $x \geq 0$, and the following condition of
``continuity at the vertex'' holds:
\be \label{kirch}
m_1 \ = \ \f {m_2} 2 \cosh^{-1} \left( \f{m_2} 4 x\right).
\ee
\end{definition}

\begin{figure}[h!]
\begin{center}
\includegraphics[width=0.75\textwidth]{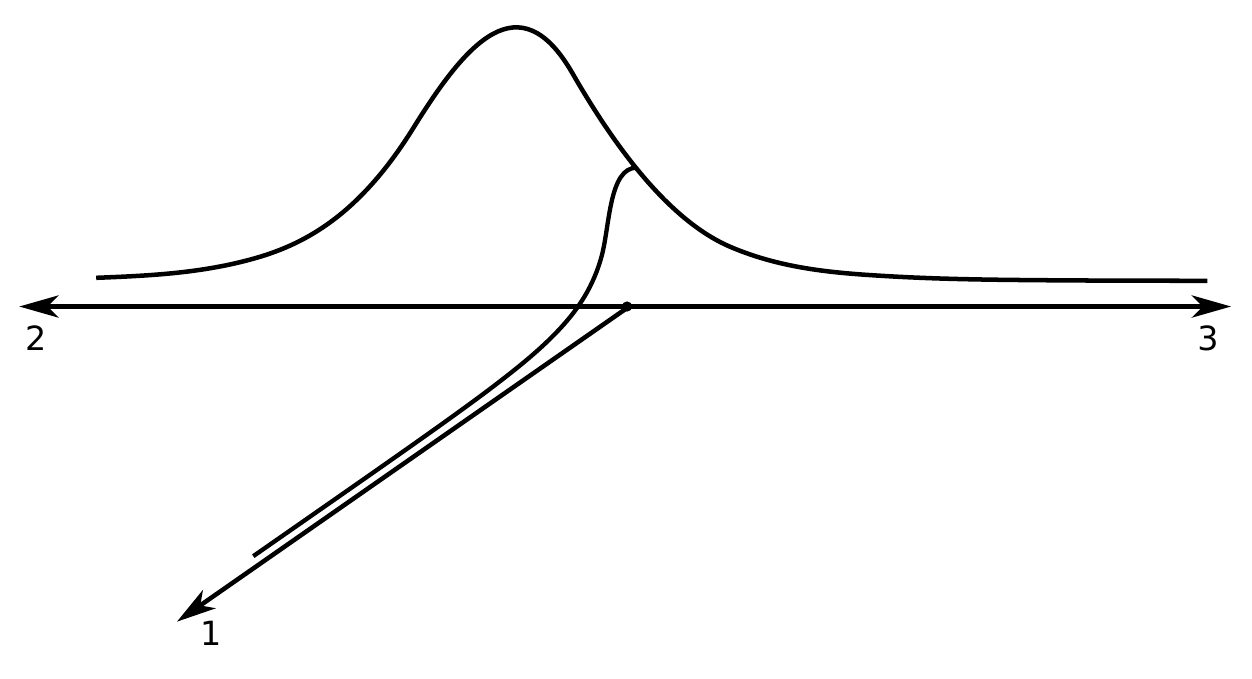}
\caption{\label{fig1} A sesquisoliton on the three-edge star graph.} 
\end{center}
\end{figure}

\n
Notice that the sesquisolitons are obviously elements of $H^1
({\mathcal G})$. As a matter of fact they also belong to $D (- \Delta
)$.   
\par\noindent
Moreover, in the case $x=0$ one has $m_1=\frac{m_2}{2}$, and one obtains a symmetric configuration with three half solitons concurring at the vertex. 
In \cite{ACFN2} it is shown that this is a standing wave for NLS equation \eqref{schrod}.
\par\noindent
Now, we define the manifold of the sequisolitons with fixed
$L^2$-norm as follows:
\be \nonumber 
{\mathcal S}_M : = \{ \Phi_{m_1, m_2}^x, \ \| \Phi_{m_1, m_2}^x \|^2 =
M \}. 
\ee
\begin{theorem} \label{result}
For any $\Psi$ such that $\| \Psi \|^2 = M$, the following 
chain holds:
\be \label{inf}
E (\Psi) > \inf_{\| \Psi \|^2 = M} E (\Psi) = \inf_{\Psi \in {\mathcal
    S}_M} E
(\Psi) = - \f {M^3} {96}.
\ee
\end{theorem}
\begin{proof}
Given $\Psi \in L^2 ({\mathcal G})$, it is possible to construct a
sesquisoliton with the same $L^2$-norm but with lower energy. We
proceed as follows. Let us suppose that 
\be \label{masses} \| \Psi_1 \| \leq \min (\| \Psi_2 \|, \| \Psi_3
\|). \ee
Then, consider the sesquisoliton \eqref{sesqui} with $m_1 =  \|
\Psi_1 \|^2, m_2 = \| \Psi_2 \|^2 + \| \Psi_3 \|^2$, and $x \geq 0$
chosen in order to satisfy the condition \eqref{kirch}. Notice that
such a choice is always possible since $2 m_1 \leq m_2$, and is unique. 

If the condition \eqref{masses} is not fulfilled, then one first
relabels the edges in order to have the minimal mass on the first one,
and thus proceeds as before.

It is immediately
seen that
$ \| \Psi_1 \|^2 = \| (\Phi_{m_1, m_2}^x )_1 \|^2 = m_1$ and 
$ \| \Psi_2 \|^2 +  \| \Psi_3 \|^2 = \| (\Phi_{m_1, m_2}^x )_2 \|^2 +
\| (\Phi_{m_1, m_2}^x )_3 \|^2 = m_2$,
thus $\| \Phi_{m_1, m_2}^x \|_2^2 = M$. 

\n
Let us define the following function on the real line:
\be \label{raddrizza} \begin{split} 
\psi (\xi) \ : = & \ \left\{ \begin{array}{c}
\Psi_2 (- \xi), \qquad \xi < 0 \\
\Psi_3 (\xi), \qquad \xi > 0
\end{array} \right.
\end{split}
\ee
and notice that, by \eqref{solit} and \eqref{sesqui}, 
\be \label{raddrizza2} \begin{split} 
\varphi_{m_2}^{-x} (\xi) \ : = & \ \left\{ \begin{array}{c}
(\Phi_{m_1,m_2}^x)_2 (- \xi), \qquad \xi < 0 \\
(\Phi_{m_1,m_2}^x)_3 (\xi), \qquad \xi > 0
\end{array} \right.
\end{split}
\ee

Furthermore, by \eqref{min1}, \eqref{min2}, \eqref{energy},
\eqref{energy1}, \eqref{energy2}, \eqref{raddrizza}, and \eqref{raddrizza2} one immediately
has the following chain of inequalities
\be \nonumber \begin{split}
E (\Psi) & \ = \ E_1 (\Psi_1) + E_2 (\psi) \ \geq \  E_1 ( \phi_{m_1} )
+  E_2 (\varphi_{m_2}^{-x})
 \ = \ E (\Phi_{m_1, m_2}^x )
\end{split} 
\ee
so the first of the two identities in \eqref{inf} is proven. To prove
the second, we
use \eqref{min1} and \eqref{min2} and obtain
\be \nonumber 
E (\Phi_{m_1, m_2}^x) \ = \ - \f {m_1^3}{24} - \f {m_2^3}{96}, 
\ee
so, noting that $M = m_1 + m_2$, we obtain
\be \label{e-sesqui-2}
E (\Phi_{m_1, m_2}^x) \ = \  - \f {m_1^3}{24} + \f {(m_1-M)^3}{96},
\ee
where $m_1$ plays the role of a parameter. We stress that, due to the
constraint of the mass of the soliton, $m_1$ can vary in the interval $(0,
M/3]$.
Differentiating \eqref{e-sesqui-2}, one immediately has that $E
(\Phi_{m_1, m_2}^x)$ is monotonically increasing in such interval, so
$$
\inf_{\Psi \in {\mathcal S}_M} E (\Psi) \ = \ \lim_{m_1 \to 0+} E (\Phi_{m_1,
  m_2}^x) \ = \ -  \f {M^3}{96}.
$$
To complete the proof, we must show that, for any $\Psi$ in
$H^1({\mathcal G})$, $E(\Psi)$ is strictly larger than $- M^3/96$.
To this aim, we notice that such an infimum cannot be
achieved, as for $m_1 =0$ the condition \eqref{kirch} does not
correspond to an admissible sesquisoliton, so it cannot be fulfilled.
\end{proof} 

\begin{cor}
The sesquisoliton $\Phi_{M/3, 2M/3}^0$ is a saddle point for the energy
functional.
\end{cor}
\begin{proof}
First, we notice that $\Phi_{M/3, 2M/3}^0$ is a critical point. Indeed,
it satisfies
the Euler-Lagrange equation for the energy functional
constrained on the manifold $\| \Psi \|^2 = M$, namely
$$
- \Delta \Psi - | \Psi |^2 \Psi + \omega \Psi \ = \ 0, 
$$
where $\omega$ is a Lagrange multiplier, coinciding with $\f
{M^2}{36}$.

In order to prove that $\Phi_{M/3, 2M/3}^0$ is a saddle point, it is sufficient to show
that it maximizes the energy restricted to a curve to which it belongs in the constraint manifold, and minimizes the energy when restricted to a different curve.

By the proof of theorem
\eqref{result} we have a curve on which $\Phi_{M/3, 2M/3}^0$ is a maximum of the energy, that is, the
curve made of sesquisolitons parametrized by $m_1$. On the other hand,
by \eqref{min1} we know that $\Phi_{M/3, 2M/3}^0$ minimizes the energy
at fixed mass on any edge. Then, it is a minimum of the energy restricted to the submanifold of
function with $\Psi_1 = \Psi_2 = \Psi_3$.
\end{proof}
\vskip 10pt

\begin{figure}[h!]
\begin{center}
\includegraphics[width=0.75\textwidth]{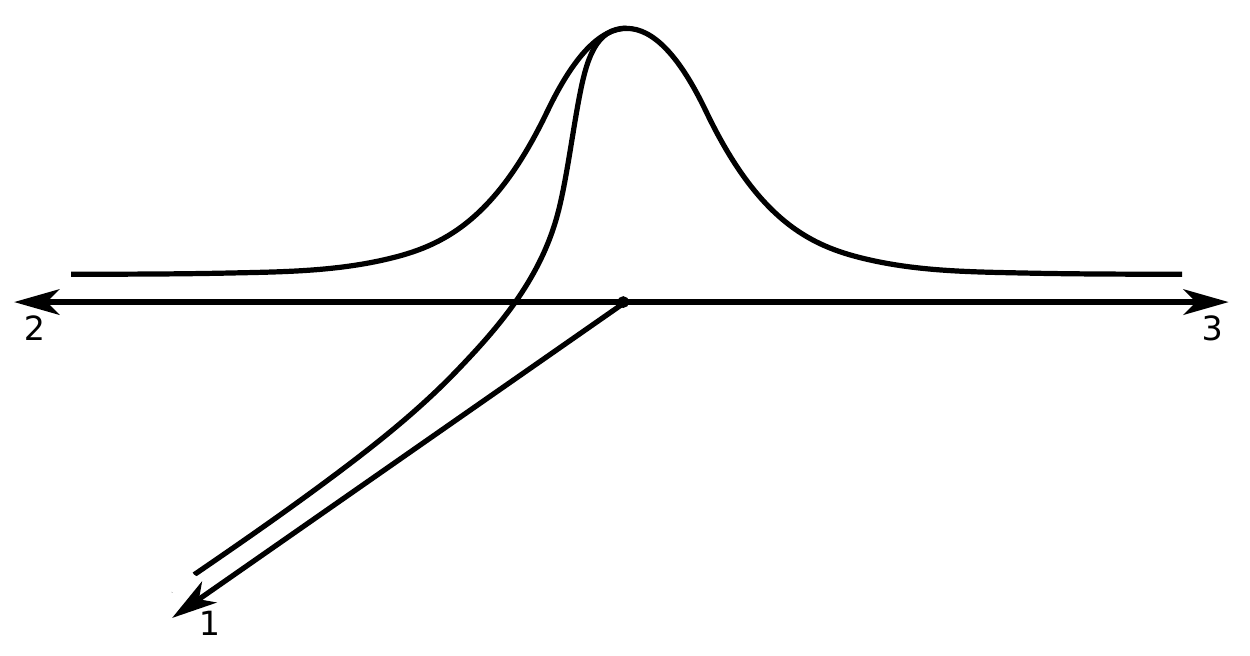}
\caption{\label{fig2}The unique stationary state of the cubic NLS on
  the three-edge star graph.} 
\end{center}
\end{figure}

\n
The previous result can be extended to all star graphs with a similar
construction, and a more systematic analysis of the character of  
stationary states on star graphs will be given in a future work.

\bigskip

\n
{\bf Acknowledgements.} R. A. is partially supported by the PRIN2009 grant 
 ``Critical Point Theory and Perturbative Methods for Nonlinear Differential Equations''.


\vskip20pt
\begin{thebibliography}{99}
\bibitem{ACFN1}Adami R., Cacciapuoti C., Finco D., Noja D.: Fast solitons on star graphs, Rev Math. Phys. {\bf 23}, 4, 409-451 (2011)
\bibitem{ACFN2}Adami, R. Cacciapuoti C., Finco D., Noja D.: Stationary states of NLS on star graphs arXiv:1104.3839v2 (2011)
\bibitem{BL1} Berestycki H., Lions P.L.: Nonlinear scalar field
  equations I and II, Arch. Rat. Mech. Anal. {\bf 82},   313-375 (1983)
\bibitem{C} Cazenave T.: Semilinear Schr\"odinger equations, AMS (2003)
\bibitem{Ex} Exner P., Keating J.P., Kuchment P., Sunada T., and
  Teplyaev A.: 
Analysis on graphs and its applications, AMS (2008).
\bibitem{Smi}Gnutzmann, S. Smilansky, U., Derevyanko S.:
 Stationary scattering from a nonlinear network, Phys. Rev. A {\bf 83}, 033831 (2011)
\bibitem{Kevr} Kevrekidis P.G., Frantzeskakis D.J., Theocharis G.,
  Kevrekidis I.G.: Guidance of matter waves through Y-junctions,
  Phys. Lett. A {\bf 317}, 51322 (2003)
\bibitem{kostrykin} Kostrykin V., Schrader R.: Kirchhoff's rule for
  quantum wires,  J.Phys. A: Math. Gen. {\bf 32} no. 4, 595-630 (1999)
\bibitem{Kuchment} Kuchment P.: Graph models for waves in thin
  structures, Waves Random Media, {\bf{12}} (4): R1-R24 (2002)
\bibitem{Sob}Sobirov Z., Matrasulov D., Sabirov K., Sawada S.,
  Nakamura K.: 
Integrable nonlinear Schr\"odinger equation on simple networks: 
 Connection formula at vertices, Phys. Rev. E {\bf 81}, 066602 (2010)
\bibitem{TOD}Tokuno A., Oshikawa M., Demler E.: Dynamics of one
  dimesional Bose  liquids: Andreev-like reflection at Y junctions and
  the absence of the  Aharonov-Bohm effect, Phys. Rev. Lett. {\bf 100}, 140402 (2008)
\bibitem{Miro}Miroshnichenko A. E., Molina M.I., Kivshar Y.S.: 
Localized modes and bistable scattering in nonlinear network
junctions,  Phys. Rev. E {\bf 75}, 046602 (2007)
\bibitem{Wan} Wan, W. Muenzel S. Fleischer, W.: Wave Tunneling and
  Hysteresis in Nonlinear  Junctions, Phys. Rev. Lett. {\bf 104}, 073903 (2010)
\end{thebibliography}
\end{document}